\documentclass[letterpaper]{article}

\usepackage[T1]{fontenc}
\usepackage{mathpazo}
\usepackage{amsmath, amssymb, amsthm}
\usepackage{booktabs}
\usepackage{xcolor}
\usepackage{graphicx}
\usepackage{enumitem}
\usepackage{geometry}
\usepackage{titlesec}
\usepackage{tikz}
\usepackage{pgfplots}

\definecolor{linkblue}{RGB}{30, 80, 180}
\definecolor{linkgreen}{RGB}{30, 120, 60}
\definecolor{linkred}{HTML}{AE0000}

\usepackage{hyperref}

\pgfplotsset{compat=1.17}

\usetikzlibrary{patterns, arrows.meta, decorations.pathreplacing, calc}

\hypersetup{
  colorlinks=true,
  urlcolor=linkblue,
  linkcolor=linkgreen,
  citecolor=linkgreen,
  filecolor=linkred,
}

\titleformat{\section}{\large\bfseries}{\thesection}{1em}{}
\titleformat{\subsection}{\normalsize\bfseries}{\thesubsection}{1em}{}

\newtheorem{definition}{Definition}
\newtheorem{theorem}{Theorem}
\newtheorem{corollary}{Corollary}
\newtheorem{proposition}{Proposition}

\title{\textbf{Engineered Simultaneity:}\\[0.3em]
{\large The Physical Impossibility of Consolidated Price Discovery\\
Across Spacelike-Separated Exchanges}}

\author{Paul Borrill\\
\small DAEDAELUS, Asilomar, CA\\
\small\texttt{paul@daedaelus.com}\\
\small ORCID: 0000-0002-7493-5189}

\date{Version 4.2 \quad 02026-NOV-24}

\begin{document}

\maketitle

\begin{abstract}
We define \emph{engineered simultaneity}: the construction of a system that
requires temporal comparison of events at spacelike-separated locations,
implements this comparison via an implicit simultaneity convention, and
represents the result as an objective measurement rather than a conventional
choice. We show that the National Best Bid and Offer (NBBO)---the
regulatory cornerstone of U.S. equity markets---is an instance of
engineered simultaneity. The NBBO requires determining ``current'' prices
across exchanges whose spatial separation places their price events
outside each other's light cones. Special relativity proves that the
temporal ordering of such events is frame-dependent: there exist inertial
reference frames in which the NBBO differs from the value reported by the
Securities Information Processor. The impossibility is not approximate; it
is exact and unavoidable within the causal structure of Minkowski
spacetime. General relativity compounds the impossibility: gravitational
time dilation introduces frame-rate discrepancies between exchanges at
different altitudes, and recent work on indefinite causal order in quantum
information theory undermines the premise of a fixed causal structure
altogether. We formalize the special-relativistic argument using the causal
precedence relation and connect it to Lamport's theorem on distributed
ordering. The operational form of the result needs no moving observers:
the NBBO depends on where it is computed, and every participant who is
not at the consolidation point holds a different one. The same
construction recurs inside a single exchange, where orders are stamped
at a perimeter and sorted by the stamp. Read as a practitioner suggests,
as a bound on the worst-case error of distributed price discovery, the
NBBO is defensible, and we propose that it be published as one. We
separate what a central market
is for (aggregating supply and demand) from what the latency race is
about (sequencing arrivals), note that races of this kind cost on the
order of \$5~billion a year in global equities, and record that in June
2026 the SEC proposed to rescind the order protection rule that makes the
NBBO binding.
\end{abstract}

\section{Introduction}

The U.S. equity market operates across 17 exchanges~\cite{sec2026},
legally domiciled in New York, New Jersey, Texas and Illinois, with their
matching engines housed almost entirely in data centers in northern New
Jersey. The NBBO applies to these equity markets only. The futures markets
that often lead them, chiefly CME Group in Aurora, Illinois, are separate,
centrally adjudicated exchanges regulated by the CFTC and outside the
NBBO~\cite{lockwood2026}. Regulation NMS, adopted
by the Securities and Exchange Commission in 2005~\cite{regnms2005},
mandates the National Best Bid and Offer (NBBO): a consolidated
quotation representing the best available bid and offer prices across
all exchanges at any given time. The order protection or trade-through rule (Rule~611)
requires trading centers to prevent trades at prices inferior to
\emph{protected quotations} displayed by other venues; best execution is
a separate duty of brokers. Protected quotations and the public NBBO have
until now been, in practice, the same thing.

On 11~June 2026 the SEC proposed rescinding Rule~611 and the prohibition
on locked and crossed markets (Rule~610(e))~\cite{sec2026}, and moved the
compliance date of its 2024 tick-size and access-fee amendments to
November 2027. Nasdaq's chief economist observed that order protection
had turned the consolidated feed into a pre-trade tape, which he argued
physics makes impossible~\cite{mackintosh2026}. The argument of this
paper is that the impossibility is exact, and where it comes from.

This regulatory architecture rests on a physical assumption: that
``at any given time'' has a well-defined meaning across geographically
separated exchanges. We show that it does not.

Einstein's special theory of relativity~\cite{einstein1905} established
that simultaneity is not an absolute relation. For events separated by a
spacelike interval---events too far apart in space for light to connect
them in the elapsed time---there is no unique answer to the question
``which happened first?'' Different inertial observers, each applying the
laws of physics correctly, assign different temporal orderings to the
same pair of events. This is not a measurement limitation. It is a
structural feature of Minkowski spacetime~\cite{minkowski1908}.

Independently, Lamport~\cite{lamport1978} proved the equivalent result for
distributed systems: without communication, no consistent ordering of
events at different nodes can be established. The ``happened-before''
relation provides the only physically meaningful ordering, and it is a
partial order---it leaves spacelike-separated events unordered.

The NBBO requires a total order. Physics provides only a partial order.
The gap between these two structures is not a technical deficiency to be
engineered away. It is a theorem.

We call the systematic practice of constructing, deploying, and
representing such a total order as if it were a physical measurement
\emph{engineered simultaneity}. The purpose of this paper is to define
the concept precisely, prove its applicability to the NBBO, and
quantify its consequences.

\section{The NBBO as a Simultaneity Claim}
\label{sec:nbbo}

\subsection{Regulatory Definition}

Regulation NMS defines the NBBO as the best bid and best offer
``currently available'' across all ``automated quotations'' of
national securities exchanges~\cite{regnms2005}. The Securities
Information Processor (SIP) receives quotation updates from each
exchange and publishes a consolidated NBBO.

The operative word is ``currently.'' For the NBBO to be well-defined,
there must exist a meaningful notion of ``current'' that applies
simultaneously across all exchanges. This requires comparing prices at
events---quote updates---occurring at different spatial locations.

\subsection{The Consolidation Process}

The SIP consolidation proceeds as follows. Each exchange $E_i$
($i = 1, \ldots, N$) generates a stream of quote updates. Each update
is an event: a price $p_i$ produced at spacetime location $(x_i, t_i)$.
The SIP, located at a fixed position $x_{\mathrm{SIP}}$, receives these
updates after propagation delays $\delta_i \geq |x_i - x_{\mathrm{SIP}}|/c$
and computes:
\begin{equation}
\mathrm{NBBO}(t) = \left(\max_i \{b_i(t)\},\; \min_i \{a_i(t)\}\right)
\label{eq:nbbo}
\end{equation}
where $b_i(t)$ and $a_i(t)$ are the bid and ask prices from exchange
$E_i$ ``at time $t$.''

But $b_i(t)$ is not the bid at exchange $E_i$ at time $t$. It is the
most recent bid that has \emph{arrived at the SIP} by time $t$. The SIP
does not observe the market; it observes its own input buffer. The NBBO
is a function of arrival times at a single location, not of emission
times at multiple locations.

\subsection{Formal Definition}

\begin{definition}[Engineered Simultaneity]
\label{def:es}
A system exhibits \emph{engineered simultaneity} when it satisfies three
conditions:
\begin{enumerate}
\item[(ES1)] It requires temporal comparison of events at
      spacelike-separated locations.
\item[(ES2)] It implements this comparison via an implicit
      simultaneity convention (typically: order of arrival at a
      designated consolidation point).
\item[(ES3)] It represents the result as an objective property of the
      system rather than as a consequence of the chosen convention.
\end{enumerate}
\end{definition}

The NBBO satisfies all three conditions. (ES1): Exchanges are spatially
separated and quote events occur at intervals shorter than light-travel
time. (ES2): The SIP imposes order by arrival time at its own location.
(ES3): Regulation NMS treats the NBBO as the objective ``best price''
for the purpose of trade-through protection.

\section{Causal Structure of the U.S.\ Exchange Network}
\label{sec:causal}

\subsection{Exchange Locations}

Whatever their legal domicile, U.S.\ equity exchanges match orders in
three data center clusters in northern New Jersey; exchanges domiciled in
Texas match there too, because that is where the participants are.
Table~\ref{tab:exchanges} lists them, with the CME futures market for
comparison, and Table~\ref{tab:latencies} gives pairwise light-time
separations.

\begin{table}[ht]
\centering
\small
\begin{tabular}{lll}
\toprule
\textbf{Matching location} & \textbf{Coordinates} & \textbf{Equity exchanges (legal domicile if not NY/NJ)} \\
\midrule
Mahwah, NJ    & 41.08$^\circ$N, 74.16$^\circ$W & NYSE, NYSE Arca, NYSE American, NYSE National, \\
              &                                & NYSE Texas (Dallas, TX) \\
Carteret, NJ  & 40.58$^\circ$N, 74.23$^\circ$W & Nasdaq, Nasdaq PSX, Nasdaq Texas (TX; formerly Nasdaq BX) \\
Secaucus, NJ  & 40.79$^\circ$N, 74.06$^\circ$W & Cboe BZX, BYX, EDGX, EDGA (Cboe: Chicago, IL)$^\dagger$; \\
              &                                & TXSE (Dallas, TX); IEX, MEMX, MIAX Pearl$^\dagger$ \\
other / unverified & ---                       & LTSE, 24X National$^\dagger$ \\
\midrule
Aurora, IL    & 41.76$^\circ$N, 88.29$^\circ$W & \emph{none}: CME Group, futures (not an NBBO venue) \\
\bottomrule
\end{tabular}
\caption{U.S.\ equity exchanges by the location of their matching engines, which is what the
physics sees, and by legal domicile where it differs. The Texas Stock Exchange matches at Equinix NY6
in Secaucus~\cite{txse2025}; NYSE Texas operates from the Mahwah data center~\cite{nysetexas2025};
Nasdaq Texas is the reincorporated Nasdaq BX. $^\dagger$Matching location not independently verified.
The count of registered equity exchanges changes as venues launch. CME is shown for the futures-equity
comparison only.}
\label{tab:exchanges}
\end{table}

\begin{table}[ht]
\centering
\small
\begin{tabular}{lrrr}
\toprule
\textbf{Pair} & \textbf{Distance (km)} & \textbf{$\Delta t_{\min}$ ($\mu$s)} & \textbf{Fiber ($\mu$s)} \\
\midrule
Mahwah--Carteret     &    56 &   186 &   280 \\
Mahwah--Secaucus     &    34 &   113 &   170 \\
Carteret--Secaucus   &    27 &    90 &   135 \\
NJ cluster--Aurora (futures) & 1,\!180 & 3,\!940 & 5,\!900 \\
\bottomrule
\end{tabular}
\caption{Pairwise separations and minimum light-time delays. $\Delta t_{\min} = d/c$ in vacuum; fiber delay assumes $n \approx 1.5$. Distances are great-circle distances between the coordinates of Table~\ref{tab:exchanges}. The last row compares equity and futures venues; it is not a pair within the NBBO.}
\label{tab:latencies}
\end{table}

\begin{figure}[ht]
\centering
\begin{tikzpicture}[scale=0.82,
  >=Stealth,
  worldline/.style={thick},
  lightcone/.style={thick, dashed, orange!80!black},
  ]

  \draw[->, thick] (-1.5,0) -- (11.5,0) node[right] {Space (km)};
  \draw[->, thick] (0,-1.8) -- (0,9) node[above left] {$ct$ (km)};

  \fill[blue!8] (-0.1,0) rectangle (0.7,8.2);
  \draw[worldline, blue!70!black] (0.3,0) -- (0.3,8.2);
  \node[above right, font=\footnotesize, blue!70!black, xshift=-2pt] at (0.7,8.2)
    {NJ cluster};
  \node[below right, font=\tiny, blue!50!black, xshift=-2pt] at (0.7,8.2)
    {NYSE, NASDAQ, CBOE};

  \draw[worldline, red!70!black] (9.0,0) -- (9.0,8.2);
  \node[above, font=\footnotesize, red!70!black] at (9.0,8.2) {CME};

  \draw[decorate, decoration={brace, amplitude=4pt, mirror}]
    (-0.1,-0.4) -- (0.7,-0.4)
    node[midway, below=5pt, font=\scriptsize] {$\lesssim\!56$\,km};

  \draw[<->, thin] (0.7,-1.2) -- (9.0,-1.2)
    node[midway, below, font=\footnotesize] {1,180\,km};

  \fill[blue!70!black] (0.3,1.5) circle (3pt);
  \node[left, font=\small, blue!70!black] at (-0.2,1.5) {$\alpha$};

  \begin{scope}
    \clip (-1.4,-0.1) rectangle (11.2,8.3);
    \draw[lightcone] (0.3,1.5) -- (7.0,8.2);
    \draw[lightcone] (0.3,1.5) -- (-1.2,3.0);
    \draw[lightcone] (9.0,4.0) -- (0.8,12.2);
    \draw[lightcone] (9.0,4.0) -- (11.2,6.2);
  \end{scope}

  \fill[red!70!black] (9.0,4.0) circle (3pt);
  \node[right, font=\small, red!70!black] at (9.2,4.0) {$\beta$};

  \begin{scope}
    \clip (-1.4,0) rectangle (11.2,8.2);
    \fill[red!12, opacity=0.5]
      (0.3,1.5) -- (5.9,7.1) -- (9.0,4.0) -- (4.35,-1.35) -- cycle;
  \end{scope}

  \node[font=\footnotesize, text=red!60!black, align=center,
        fill=white, fill opacity=0.7, text opacity=1, inner sep=2pt,
        rounded corners=1pt]
    at (6.5,5.8)
    {Spacelike region\\[-1pt]{\scriptsize ordering is frame-dependent}};

  \draw[<->, thin, green!50!black] (0.5,2.5) -- (8.8,2.5)
    node[midway, above, font=\footnotesize, fill=white, inner sep=2pt]
    {$\Delta t_{\min} = 3{,}940\;\mu$s};

  \fill[black!60!green] (0.3,5.5) circle (2.5pt);
  \node[left, font=\footnotesize, black!60!green] at (-0.2,5.5)
    {SIP receives $\alpha$};

  \fill[black!60!green] (0.3,7.2) circle (2.5pt);
  \node[left, font=\footnotesize, black!60!green] at (-0.2,7.2)
    {SIP receives $\beta$};

  \draw[->, thick, green!50!black, dashed] (0.3,7.2) -- (2.0,7.2);
  \node[above right, font=\footnotesize, green!50!black, yshift=-1pt] at (1.0,7.2)
    {NBBO computed};

\end{tikzpicture}
\caption{Minkowski diagram of the U.S.\ exchange network. Events
$\alpha$ (NYSE quote) and $\beta$ (CME quote) are
spacelike-separated: no light signal from $\alpha$ can reach $\beta$
before $\beta$ occurs. Their temporal ordering is frame-dependent.
CME futures quotes are not inputs to the NBBO; the pair is drawn at
corridor scale for legibility, and because futures-to-equity updates
are the information that latency arbitrage races on. The NBBO case is the
same geometry inside the New Jersey cluster, at up to 56\,km and
186\,$\mu$s (Table~\ref{tab:latencies}).
The SIP, located near the NJ cluster, receives both events and
computes an NBBO---but this NBBO reflects arrival order at the SIP,
not any frame-independent temporal ordering.}
\label{fig:minkowski}
\end{figure}

\subsection{Minkowski Analysis}

Consider two exchanges $E_A$ at position $x_A$ and $E_B$ at position
$x_B$, with $d = |x_A - x_B|$. A quote event $\alpha$ at $E_A$
occurring at coordinate time $t_\alpha$ and a quote event $\beta$ at
$E_B$ at $t_\beta$ are spacelike-separated if and only if:
\begin{equation}
|t_\alpha - t_\beta| < \frac{d}{c}
\label{eq:spacelike}
\end{equation}
For spacelike-separated events, the Lorentz transformation yields:
\begin{equation}
t'_\alpha - t'_\beta = \gamma\left[(t_\alpha - t_\beta) - \frac{v}{c^2}(x_\alpha - x_\beta)\right]
\label{eq:lorentz}
\end{equation}
where $v$ is the velocity of the primed frame and
$\gamma = (1 - v^2/c^2)^{-1/2}$. When events are spacelike-separated,
there exist frames in which $t'_\alpha > t'_\beta$ and frames in which
$t'_\alpha < t'_\beta$. The sign of the temporal interval---which event
is ``first''---depends on the observer.

\begin{theorem}[Frame-Dependence of the NBBO]
\label{thm:main}
Let $E_A$ and $E_B$ be exchanges separated by distance $d > 0$. Let
$\alpha$ and $\beta$ be quote events at $E_A$ and $E_B$ respectively,
with $|t_\alpha - t_\beta| < d/c$. Then there exist inertial frames
$S$ and $S'$ such that the NBBO computed from $\{\alpha, \beta\}$
differs between $S$ and $S'$.
\end{theorem}

\begin{proof}
Suppose $\alpha$ reports bid $b_A$ and $\beta$ reports bid $b_B > b_A$.
In frame $S$, where $\alpha$ occurs before $\beta$, the NBBO at any
time after both events reflects $b_B$. In frame $S'$, where $\beta$
occurs before $\alpha$, the sequence of NBBO updates differs: $b_B$
enters the NBBO first, followed by no change (since $b_A < b_B$). More
critically, at times between the two events in each frame, the NBBO
disagrees: frame $S$ reports $b_A$ as best bid while $b_B$ has not yet
occurred; frame $S'$ already reports $b_B$.

Since the NBBO is a function of which quote events are deemed
``current,'' and ``current'' is frame-dependent for spacelike-separated
events, the NBBO is frame-dependent. \qed
\end{proof}

\begin{corollary}
\label{cor:convention}
The NBBO reported by the SIP is not a measurement of the market state.
It is an artifact of a particular simultaneity convention: the
convention that ``current'' means ``received at $x_{\mathrm{SIP}}$ by
coordinate time $t$.'' Different conventions yield different NBBOs, and
no convention is privileged by physics.
\end{corollary}

\subsection{Vantage, Not Velocity}
\label{sec:vantage}

Theorem~\ref{thm:main} uses moving frames, and a reader may object that
the SIP sits at one place: the order in which events arrive at a single
point is the same in every inertial frame. The objection is correct, and
it locates the argument precisely. The SIP's \emph{computation} is
frame-invariant. What is not invariant is the claim that its output is
\emph{the} state of the market, which is condition (ES3).

The operational version needs no relativity of motion at all.

\begin{proposition}[Vantage dependence]
\label{prop:vantage}
Let observers $O_1$ and $O_2$ be at rest at different locations, each
computing Equation~(\ref{eq:nbbo}) from the quotes that have arrived at it.
If two quote events $\alpha$, $\beta$ satisfy
$|t_\alpha - t_\beta| < \big|\,|x_\alpha - x_{O_1}| - |x_\alpha - x_{O_2}|\,\big|/c$
for a suitable ordering of the terms, then $O_1$ and $O_2$ hold
different NBBOs over an interval of time, and neither is mistaken.
\end{proposition}

\noindent The proof is the difference of two light delays. The case is
not hypothetical. The consolidated quote for Nasdaq-listed stocks is
computed in one place, Carteret~\cite{mackintosh2025}; a direct-feed
subscriber in Mahwah or Secaucus holds a different NBBO, and so does every
broker that builds its own. The SEC's 2020 market-data rules, which allow
competing consolidators, would make several official NBBOs coexist by
design. Relativity explains why no engineering can remove the vantage
dependence. The vantage dependence is what a market participant meets.

The constructive consequence is developed in a companion
paper~\cite{borrill2026adjudication}: a matching engine ``does not, in
fact, rely on global time; it adjudicates by arrival order at a single
point, which is a total, local, frame-invariant order. The defect is its
singularity, not the principle.''

\section{Connection to Distributed Systems Theory}
\label{sec:lamport}

Lamport's foundational result~\cite{lamport1978} establishes that in a
distributed system, the only physically meaningful ordering of events is
the \emph{happened-before} relation $\rightarrow$, defined by:

\begin{enumerate}
\item If $a$ and $b$ are events at the same process and $a$ occurs
before $b$, then $a \rightarrow b$.
\item If $a$ is the sending of a message and $b$ is its receipt, then
$a \rightarrow b$.
\item The relation is transitive.
\end{enumerate}

Events not related by $\rightarrow$ are \emph{concurrent}: they have no
physically meaningful ordering. Lamport noted explicitly that this
structure is isomorphic to the causal structure of special
relativity~\cite{lamport1978}. The happened-before relation is the
discrete analogue of the timelike/lightlike ordering in Minkowski
spacetime; concurrent events are the analogue of spacelike-separated
events.

\begin{proposition}
\label{prop:lamport}
Quote events at spacelike-separated exchanges that have not yet been
communicated to any common process are concurrent in Lamport's sense.
The NBBO requires ordering these concurrent events. No such ordering
can be derived from the happened-before relation.
\end{proposition}

The SIP resolves this by fiat: it orders concurrent events by arrival
time at a single location. This is a valid synchronization convention
but not a unique one. A SIP located elsewhere, or using a different
tie-breaking rule, would produce a different NBBO.

Google's Spanner database~\cite{corbett2012} confronted the same
impossibility. Its solution---TrueTime---represents time as an
\emph{interval of uncertainty} and waits out the uncertainty before
declaring a transaction committed. The financial system, by contrast,
pretends the uncertainty does not exist.

\section{Quantifying the Exploitation}
\label{sec:exploitation}

The frame-dependence of the NBBO creates an exploitable asymmetry.
Participants who receive exchange data via direct proprietary feeds
observe quote events on the order of tens of microseconds after
emission~\cite{ding2014}; participants relying on the consolidated SIP
observed the same events approximately
$1{,}128\;\mu\mathrm{s}$ after emission in Bartlett and McCrary's
2015--2016 sample~\cite{bartlett2019}---a ratio exceeding $50{:}1$. SIP
latency has since fallen; the CTA plan now publishes realized latency
from the matching-engine event to SIP dissemination~\cite{ctalatency}.
The argument does not depend on the size of the ratio, only on its being
nonzero.

\begin{figure}[ht]
\centering
\begin{tikzpicture}
\begin{axis}[
    width=0.88\textwidth,
    height=7.5cm,
    ymode=log,
    ymin=1, ymax=100000,
    xmin=0, xmax=1300,
    xlabel={Inter-exchange distance (km)},
    ylabel={Time ($\mu$s)},
    grid=major,
    grid style={gray!20},
    legend style={at={(0.03,0.97)}, anchor=north west, font=\footnotesize,
                  draw=gray!50, fill=white, fill opacity=0.95,
                  cells={anchor=west}},
    ]

    \addplot[fill=red!8, draw=none, forget plot, domain=1:1300, samples=50]
      {x / 0.2998} \closedcycle;

    \addplot[domain=1:1300, samples=100, thick, orange, dashed]
      {x / 0.2998};
    \addlegendentry{Light-time $d/c$ (vacuum)}

    \addplot[domain=1:1300, samples=100, thick, orange!60!brown,
             dash pattern=on 4pt off 2pt on 1pt off 2pt]
      {x / 0.1999};
    \addlegendentry{Fiber propagation ($n \!=\! 1.5$)}

    \addplot[domain=0:1300, samples=2, thick, red] {1128};
    \addlegendentry{SIP consolidated feed (${\sim}1{,}128\;\mu$s)}

    \addplot[domain=0:1300, samples=2, thick, blue] {20};
    \addlegendentry{Direct feed ($\mathcal{O}(10)\;\mu$s)}

    \addplot[domain=0:1300, samples=2, thick, black!60!green] {10};
    \addlegendentry{HFT execution ($<\!10\;\mu$s)}

    \addplot[only marks, mark=*, mark size=3pt, blue!70!black]
      coordinates {(56, 186)};
    \node[font=\scriptsize, anchor=south west, xshift=2pt] at (axis cs:56,186)
      {Mahwah--Carteret};

    \addplot[only marks, mark=*, mark size=3pt, red!70!black]
      coordinates {(1180, 3940)};
    \node[font=\scriptsize, anchor=south east, xshift=-2pt] at (axis cs:1180,3940)
      {NJ--Aurora};

    \draw[<->, thick, gray] (axis cs:1220,20) -- (axis cs:1220,1128);
    \node[font=\scriptsize, text=gray, rotate=90, anchor=south]
      at (axis cs:1240,120) {$>$50$\times$ gap};

    \node[font=\footnotesize, text=red!50!black, rotate=22]
      at (axis cs:650,40)
      {Spacelike (no causal connection)};

\end{axis}
\end{tikzpicture}
\caption{Timescales in the U.S.\ exchange network. The orange dashed
line marks the light-time separation $d/c$: events closer in time than
this are spacelike-separated. HFT execution times ($<$10\,$\mu$s) and
direct feed latencies ($\mathcal{O}(10)\,\mu$s) operate far below the
light-time for the NJ--Chicago corridor (3,940\,$\mu$s), while even
within the NJ cluster (186\,$\mu$s) they approach the boundary. The SIP and $>$50$\times$ figures are from Bartlett and McCrary's 2015--2016 sample. The
$>$50$\times$ gap between direct feeds and the SIP consolidated feed
(right edge annotation) defines the exploitable asymmetry.}
\label{fig:latency}
\end{figure}

This is not a speed advantage in the conventional sense. It is a
\emph{frame advantage}: direct-feed participants effectively inhabit a
reference frame closer to the exchange's local frame, while SIP
participants inhabit a frame displaced by the consolidation delay. The
two frames assign different temporal orderings to recent quote events,
and therefore compute different NBBOs.

Aquilina, Budish, and O'Neill~\cite{aquilina2022} measure the
consequences directly. Using message data from the London Stock
Exchange, which record failed attempts as well as successes, they observe
latency arbitrage ``races''---several participants reacting to the same
public information by sending orders to the same venue within
microseconds---approximately once per minute per FTSE~100 stock. The
modal race lasts 5--10\,$\mu$s. The tax is about 0.42 basis points
of all volume, the sums at stake are on the order of \$5~billion a year in
global equity markets, and removing the races would lower the cost of
liquidity by about 17\%. A related study of U.K.\ dark pools~\cite{aquilina2023}
finds that high-frequency traders are on the profitable side of trades
at stale reference prices between 96 and 99\% of the time.

Two points of attribution matter. The London Stock Exchange has no NBBO:
these are races at a single venue, decided by order of arrival at its
matching engine. The \$5~billion is therefore the cost of
\emph{arrival-order adjudication in continuous time}, of which the NBBO
is one regulatory form, not a sum extracted through the NBBO alone. The
NBBO-specific channel is the one Bartlett and McCrary measure, and it is
small (Section~\ref{sec:magnitude}). Both findings support the thesis:
the rent attaches to a privileged point of arrival, wherever it is.

Budish et al.~\cite{budish2015} characterize this extraction as
socially wasteful: it represents private profit from temporal ambiguity,
not compensation for liquidity provision or price discovery.

\subsection{The Magnitude Counterargument}
\label{sec:magnitude}

Bartlett and McCrary~\cite{bartlett2019} report that the SIP-reported
NBBO matches the ``true'' NBBO 97\% of the time, with aggregate gains
of only \$14.4~million on \$3.7~trillion in trading volume. Two
responses are warranted.

First, Bartlett and McCrary's ``true'' NBBO is itself computed using
direct-feed data at a single location---another simultaneity convention.
They have not measured the gap between convention and reality; they have
measured the gap between two conventions.

Second, the argument from magnitude misidentifies the claim. The
impossibility result (Theorem~\ref{thm:main}) does not depend on how
often the NBBO is ``wrong'' or by how much. It depends on whether the
NBBO is \emph{well-defined}. A thermometer that reads correctly 97\% of
the time but is calibrated to a temperature scale based on phlogiston
theory is not approximately right. It is conceptually incoherent. The
3\% where it fails is not noise; it is where the incoherence becomes
visible. The larger cost, per Aquilina et al., is paid wherever arrival
order at a single point decides who trades.

\section{The Pit and the Perimeter}
\label{sec:pit}

The NBBO is not the only engineered simultaneity in a market. The same
construction recurs inside a single exchange, at a scale of meters.

The account that follows is due to John Lockwood of Algo-Logic Systems,
Inc., given at a meeting of the Mulligan Stew Gang in September
2026.

In the open-outcry pits of the Chicago futures exchanges, traders for each
instrument stood around one ring and bid by voice and hand signal. Every
participant could hear and see every other; the ring was the market's one
point of adjudication, and when two traders bid the same price at the
same moment, the tie was settled by who was seen or heard first. The New
York Stock Exchange ran a similar floor for equities; U.S.\ equity trading
was later divided among many competing exchanges in different
locations~\cite{lockwood2026}, and the NBBO is what stitches their
separate points of adjudication back into one quotation. Each electronic
exchange preserved the single point and replaced the ring with
a \emph{perimeter}. Orders are stamped with a nanosecond timestamp as they
cross the exchange's boundary, from clocks disciplined by PTP to a common
reference; they then cross several stages of internal switching, which do
not preserve arrival order, and the matching engine sorts them by the
perimeter stamp before matching. Participants' fiber connections within the
co-location facility are cut to equal length so that no customer is nearer
the perimeter than another. In Lockwood's words, the order is
``time stamped, but it doesn't use timestamps.''

Read against Definition~\ref{def:es}, the perimeter satisfies all three
conditions at the scale of a building. Orders arriving at different
perimeter ports within a few hundred nanoseconds are spacelike-separated
(ES1); the exchange orders them by a convention, the reading of
synchronized clocks at the boundary (ES2); and the result is treated as
the fact of who was first (ES3). The clock is not part of the matching
rule, which is first-in-first-out; it is the mechanism by which a
first-in-first-out decision survives an interior that does not keep
order. The exchange pays continuously, in clocks and in cable, for a
property its own network does not provide: order preserved, and known to
be preserved.

\subsection{What the Hub Is For}
\label{sec:hub}

None of this argues against central markets. A hub earns its place by
\emph{aggregation}: it puts buyers and sellers in one book, so that one
price clears supply and demand that would otherwise meet only locally.

The following example is due to Lockwood. An Anchorage fisherman who sells only to his neighbors gets a
low price, because cod is plentiful in Alaska; a New York restaurant that
buys only from New Yorkers pays a high one, because there is none. An
exchange in Chicago, or an agent there acting for the fisherman, discovers
one price for the two, and both do better, because the price more fairly
reflects supply and demand across a geographically diverse
market~\cite{lockwood2026}.

The latency race concerns something else: \emph{sequencing}, the order in
which participants reach the hub. The example depends on aggregation and
not at all on microseconds; the fish travel for days, and light crosses
from Anchorage to New York in 18\,ms. Chicago is also nearly on that path
(routing through it adds 6\% to the great-circle distance), whereas a
Houston--Birmingham trade routed through Aurora travels 2.6 times the
direct distance. Aggregation and sequencing can be separated in practice:
frequent batch auctions keep the single book and remove the
race~\cite{budish2015}. The critique of this paper is aimed at
sequencing, and at the claim that the order it produces is a physical
fact.

Commenting on a draft of this paper, Lockwood made the point from the
practitioner's side: consolidating the exchange to operate as a single
point in space makes it more fair, not less, and the key is to put the
trading agents as close to the matching engine as
possible~\cite{lockwood2026}. We agree, and the point sharpens the claim
rather than opposing it. A single point is fair among those who reach it
equally, which is exactly what equal-length cross-connects inside the
co-location facility buy. The price of that fairness is that proximity to
the point becomes the scarce good: the co-location premium, and the race
to the perimeter, are what it costs. The question this paper raises is not
whether adjudication should happen at a point, but whether nearness to one
privileged point should be what the market sells, or whether, as the
companion paper argues~\cite{borrill2026adjudication}, every participant
can be given a point of its own.

\subsection{When Simultaneity Was Safe}
\label{sec:chi}

Why did the pit work as a single point, and why must every later design
work so hard to recreate it, when the hub's real contribution is
aggregation? A dimensionless ratio answers both. Let
$\tau_d$ be the time a participant needs to decide and act, and $\tau_c$
the time a signal needs to cross the market. Define
\begin{equation}
\chi \;=\; \frac{\tau_d}{\tau_c}.
\label{eq:chi}
\end{equation}
When $\chi \gg 1$, every participant's message arrives well inside every
other participant's reaction window. Which of two messages reached the
point first then carries nothing either sender could have acted on, and
treating the market as sharing one ``now'' costs nothing. When
$\chi \lesssim 1$, the race is decided by geometry, and the order of arrival
is a property of where the point of adjudication was put.
Table~\ref{tab:chi} places the regimes; the figures are orders of magnitude
under the stated assumptions.

\begin{table}[h]
\centering
\small
\begin{tabular}{@{}lccc@{}}
\toprule
Setting & $\tau_d$ & $\tau_c$ & $\chi$ \\
\midrule
Pit, by sight (15--30\,m, light) & 0.2--0.5\,s & 50--100\,ns & $10^{6}$--$10^{7}$ \\
Pit, by voice (15--30\,m, sound) & 0.2--0.5\,s & 44--87\,ms & 2--11 \\
Human at a screen, New York--San Francisco & 0.2--0.5\,s & 20--40\,ms & 5--25 \\
Co-location hall (50--200\,m fiber) & 14\,ns--1\,$\mu$s & 0.25--1\,$\mu$s & $10^{-2}$--4 \\
NJ cluster, Mahwah--Carteret (vacuum) & 14\,ns--1\,$\mu$s & 187\,$\mu$s & $10^{-4}$--$10^{-2}$ \\
NJ cluster--Aurora (vacuum) & 14\,ns--1\,$\mu$s & 3.94\,ms & $10^{-6}$--$10^{-4}$ \\
NJ cluster with a 10--100\,ms batch & $w$ = 10--100\,ms & 187\,$\mu$s & 50--500 \\
\bottomrule
\end{tabular}
\caption{The ratio $\chi$ of Eq.~(\ref{eq:chi}) across market settings.
Human decision times are assumed; the machine figure runs from the
fastest published tick-to-trade result, 13.9\,ns~\cite{amd2024stac}, to
1\,$\mu$s; distances are from Table~1. The last row uses the batch
interval $w$ in place of $\tau_d$.}
\label{tab:chi}
\end{table}

Two features of the pit show that the single point was approximate even
there. First, it had two channels. Hand signals travel at the speed of
light; voice travels at the speed of sound, and across a pit its $\chi$ is
of order a few. By the exchange's own account, hand signals were preferred
because the noise on the floor made voice impractical~\cite{cmehandsignals},
and the noise level itself carried information that traders traded
on~\cite{coval2001}. Second, the pit had a written protocol for the case in
which its simultaneity failed. An \emph{outtrade}, two records of one trade
that disagree on price, quantity, month or side, was reconciled by the two
traders themselves: notify the counterparty, agree who covers, cover, no
later than the opening of the next regular session, with liability capped
at that opening price. Where the two disagreed on price, a broker's
recorded price prevailed over a local's, and between two brokers or two
locals the buyer's price cleared~\cite{cme2008outtrades}. At the place in
market history where simultaneity was closest to true, the market still
needed a bilateral reconciliation completed by a convention (ES2).

Electronic trading moved the machines across $\chi = 1$, not the people. A
human reading a screen across the continent still sits near $\chi \approx
10$, and for that reader a consolidated quotation is a faithful summary. A
machine that reacts in tens of nanoseconds across the 56\,km of the New
Jersey cluster sits near $\chi \approx 10^{-4}$. The simultaneity claim of
the NBBO fails for exactly the participants who trade on it fastest.

The ratio also sorts the remedies. There are three ways to restore
$\chi \gg 1$ or to do without it. (i)~\emph{Shrink $\tau_c$}: bring
participants to one point and equalize their paths, which is the perimeter
and co-location of Section~\ref{sec:pit} and Lockwood's
prescription~\cite{lockwood2026}. It works inside one building and makes
proximity the scarce good. (ii)~\emph{Raise $\tau_d$}: a frequent batch
auction with interval $w$ treats all orders inside $w$ as simultaneous, so
that $\chi_{\mathrm{eff}} = w/\tau_c$; a 100\,ms interval over the New
Jersey cluster gives $\chi_{\mathrm{eff}} \approx 500$, the pit's regime
again~\cite{budish2015}. A seller's priced waiting window is the same move
made one seller at a time~\cite{borrill2026reservation}. (iii)~\emph{Remove
the need}: adjudicate at the node that commits~\cite{borrill2026adjudication},
so that no ordering decision has to cross the network; $\chi$ then sets the
price of waiting rather than the correctness of the order. The ratio and
remedy~(iii) are the author's analysis, not Lockwood's.

\section{Discussion}
\label{sec:discussion}

Engineered simultaneity is a \emph{category mistake} in the sense of
Ryle~\cite{ryle1949}: it applies a concept (absolute simultaneity) to a
domain (spacelike-separated events) where the concept has no meaning. The
mistake is not that the SIP computes something---it does. The mistake is
that the result is presented as ``the best price'' rather than as ``the
best price according to one particular, non-unique simultaneity
convention.''

The situation is precisely analogous to the pre-relativistic assumption
of absolute time. Newtonian mechanics assumes a universal ``now'' that
all observers share. Special relativity showed this to be a convention,
not a fact. Financial regulation has not yet absorbed this lesson.

Lewis~\cite{lewis2014} popularized the idea that markets were ``rigged''
by speed differentials and could be fixed by leveling the temporal playing
field. The implicit assumption was that there exists a meaningful notion
of a simultaneous market state that everyone could, in principle, reach.
That assumption is incorrect. \textit{Flash Boys} correctly identified
latency arbitrage but misdiagnosed its root cause. The problem is not
insufficient synchronization. The problem is the belief that
synchronization can create simultaneity in the first place.

Four structural observations follow:

\textit{First}, the impossibility is exact. It does not depend on
clock accuracy, synchronization protocols, or future technological
improvements. No technology can make spacelike-separated events timelike.
The speed of light is not an engineering constraint; it is a law. In
practice, as a practitioner put it in commenting on this paper, there is
no exact way to timestamp spacelike-separated events, so they are
timestamped approximately, relative to a clock at one central
location~\cite{lockwood2026}. That is sound engineering; the error lies
only in treating the approximation as exact.

\textit{Second}, the NBBO is not a measurement but a \emph{declaration}.
Timestamps assigned by the SIP do not record when events happened in any
frame-independent sense. They record when events were registered under a
particular convention. Any system that treats these timestamps as
objective measurements of temporal order is committing a category
mistake.

\textit{Third}, the impossibility is not academic. It has measurable
economic consequences, though not through the NBBO alone. Latency rents
are earned by traders acting on information that others receive
later~\cite{lockwood2026}: participants co-located with matching engines
and reading direct feeds act before those who rely on consolidated data.
The NBBO does not create that asymmetry; it is the place where the
asymmetry is written down as if it were the market's state.

\textit{Fourth}, the impossibility extends beyond special relativity.
General relativity further undermines the premise: clocks at different
gravitational potentials tick at different rates---an effect actively
corrected in GPS~\cite{ashby2003}. HFT infrastructure spans buildings,
elevations, thermal gradients, and mechanical environments. Nanosecond
agreement is already a negotiated approximation layered atop unprovable
assumptions about symmetry and stability~\cite{rindler2006}. The idea
of a single, authoritative market clock is physically indefensible.

\subsection{The NBBO as an Error Bound}
\label{sec:errorbound}

Lockwood's summary of the NBBO is that it is not the cause of unfairness
but a means to bound the worst-case error in distributed price
discovery~\cite{lockwood2026}. We accept that reading, and it locates
the remedy. A bound is honest when it is stated as a bound. The NBBO's
defect is condition (ES3) alone: a quotation computed at one place, from
inputs of differing age, is presented as a point value with no error bar.
Spanner's TrueTime~\cite{corbett2012} shows the alternative: publish the
interval, and let the user wait out or price the uncertainty.

Asked what alternative would improve price discovery in a distributed
network, we offer three, in increasing order of departure from the
present design:
\begin{enumerate}
\item \textbf{State the bound.} Publish each consolidated quotation with
the age of each input and the vantage at which it was computed, so that
a quote is a value with an interval, not a value alone.
\item \textbf{Batch the race away, keep the hub.} Frequent batch
auctions~\cite{budish2015} preserve aggregation at one point and treat
all orders in an interval as simultaneous, which removes the value of
arriving first inside it.
\item \textbf{Relocate adjudication.} Give each participant a point of
its own~\cite{borrill2026adjudication}, or let each seller adjudicate
sealed bids at its own clock, so that the cost of distance is priced in
advance rather than taken by the fastest~\cite{borrill2026reservation}.
\end{enumerate}

\subsection{Price Discovery Without a Point}
\label{sec:discovery}

Remedy~(iii) of Section~\ref{sec:chi} meets the strongest objection a
practitioner raises against it. A bilateral commit between one buyer and
one seller yields \emph{a} price, agreed by two parties. The cod example of
Section~\ref{sec:hub} needs \emph{the} price, one that reflects supply and
demand across a geographically diverse market~\cite{lockwood2026}. A hub
supplies that by construction; a network of many adjudicators does not.

Our answer is that a distributed market discovers a common price the way
the U.S.\ equity market already does across its venues: participants who
can see more than one adjudicator trade between them until the prices
agree, and each venue's contribution to the common price can be
measured~\cite{hasbrouck1995}. The agents who do this are arbitrageurs, and
the fastest of them are latency arbitrageurs. In a distributed design,
arbitrage is therefore not the disease; it is the mechanism of convergence.
The sums measured by Aquilina, Budish and O'Neill~\cite{aquilina2022} are
better read as the price of that convergence, and the design question is
who pays it, and whether it is priced in advance or taken by the fastest.
Lockwood's own description of the NBBO, a bound on the worst-case error in
distributed price discovery~\cite{lockwood2026}, accepts the premise: across
spacelike-separated venues price discovery is distributed, and what remains
open is how honestly its error is stated and how cheaply it is closed.

\subsection{Implications for Time Infrastructure}

The analysis presented here was delivered as an invited talk at the
Open Compute Project Time Appliances Project (OCP-TAP) on January~28,
2026~\cite{borrill2026ocp}. OCP-TAP builds precision time infrastructure
for data centers and distributed systems. The relevance is direct:
time appliances measure, but they do not define reality. Synchronization
has limits that must be made explicit. Infrastructure should expose
uncertainty, not hide it~\cite{borrill2026ocpslides}. A companion
paper~\cite{borrill2026fiction} traces the fiction of synchronized time
from its cultural origins through daylight saving time and leap seconds
to its consequences in distributed computing and financial markets.

\subsection{Beyond Relativity: Indefinite Causal Order}

The impossibility argument can be strengthened further. In quantum
mechanics, time is not an observable---there is no time operator.
Time is a parameter external to the system, inferred from interaction
rather than observed as a conserved quantity~\cite{pauli1958}.
Protocols that treat timestamps as ground truth rather than metadata
are therefore already on shaky physical footing.

Modern quantum information theory goes further. Work by Oreshkov,
Costa, and Brukner~\cite{oreshkov2012} shows that some processes have
no definite causal order at all. Event $A$ is not before or after event
$B$ in any global sense. This is not philosophy; it is experimentally
realizable physics. As markets become faster, more reflexive, and more
tightly coupled, their behavior increasingly resembles systems where
causal order is local and observer-dependent rather than globally
defined.

\subsection{The Structure of the Advantage}

A system is functionally fraudulent when it claims physical
legitimacy it cannot have, exploits asymmetries invisible to most
participants, and converts uncertainty into private profit. Engineered
simultaneity satisfies all three conditions. The HFT industry
understands that simultaneity is not physically real, yet markets the
appearance of it as neutral, objective, and fair. The real advantage
is not speed---speed differences are visible and competitive. The
real advantage is control over the \emph{definition} of ``now'':
control over where time is measured, when clocks are sampled, and whose
clock counts. This is not about speed. It is about
time-definition control, and most market participants cannot audit it.

The counterargument deserves its due. Intermediaries who stand between
buyers and sellers are as old as markets---the scalper who buys concert
tickets at face value, the airline that prices every seat differently,
the roadside vendor who sells ice cream to a traffic jam at three times
the shop price---and high-frequency market makers provide quotes that
narrow spreads for everyone. The distinction drawn here is narrower: an
advantage obtained by reaching a privileged point of arrival first, and
presented as the market's objective state, is a property of the
convention, not of the goods or of the service. That is the part a
different market design can remove.

We note that Wissner-Gross and Freer~\cite{wissnergross2010} identified
the relevance of relativistic constraints to financial markets in 2010,
Angel~\cite{angel2014} argued for regulatory engagement with
relativistic constraints, Laughlin, Aguirre, and
Grundfest~\cite{laughlin2014} documented the 3\,ms latency reduction
on the Chicago--New York corridor driven by microwave infrastructure,
and Buchanan~\cite{buchanan2015} extended the analysis. The present work differs in formalizing the impossibility
as a theorem about the causal structure of the exchange network, rather
than as a constraint on arbitrage strategies, and in introducing the
concept of engineered simultaneity as a precise descriptor for the
mechanism.

\section{Conclusion}

The NBBO assumes a spacelike hypersurface of simultaneity that special
relativity proves cannot exist. General relativity compounds the problem
through gravitational time dilation, and quantum information theory
further undermines the premise through indefinite causal order. We have
formalized the special-relativistic impossibility as
Theorem~\ref{thm:main}, connected it to Lamport's impossibility of
distributed ordering, and argued that arrival-order adjudication at a
privileged point, of which the NBBO is the national form, carries costs of
the order of \$5~billion a year in global equities. Read as a bound on the
worst-case error of distributed price discovery, the NBBO is defensible;
it should be published as one (Section~\ref{sec:errorbound}). Rule~611, which
made that point binding across U.S.\ venues, may now be rescinded.

The concept of \emph{engineered simultaneity}---the systematic
construction and deployment of systems that assume, implement, and
represent impossible temporal comparisons as objective---applies beyond
financial markets to any distributed system that treats timestamps as
physical measurements of ordering. We submit that recognizing this
category mistake is a prerequisite for principled system design.

The constructive alternative is developed in \emph{Adjudication at the
Worldline}~\cite{borrill2026adjudication}: relocate adjudication ``from
that one conventional oracle to a local, physical oracle at each
participant's own worldline, so that every node is the centre of its own
light cone and simultaneity is constructed by a completed bilateral
handshake rather than measured against a shared clock.'' A design note
works one form of this through, a seller-rooted sealed-bid market in
which arrival order inside a batch carries no rent and the cost of
distance is priced rather than raced~\cite{borrill2026reservation}.

A further companion paper develops the broader theoretical framework---the
Forward-In-Time-Only (FITO) category mistake in distributed
systems~\cite{borrill2026fito}---while a second traces why
synchronized time is itself a fiction, from its cultural origins
through daylight saving time and leap seconds to its consequences
in financial markets~\cite{borrill2026fiction}.

\section*{Acknowledgments}

The author thanks John Lockwood (Algo-Logic Systems, Inc.) for the account
of the pit and the perimeter in Section~\ref{sec:pit}, the example in
Section~\ref{sec:hub}, the objection answered in Section~\ref{sec:discovery}, and a careful annotated review of a draft, and the members of the
Mulligan Stew Gang for the conversations in which they arose.

\end{document}